\documentclass[runningheads]{llncs}
\usepackage{graphicx}
\usepackage{amsmath}
\usepackage{amssymb}
\usepackage{multirow}
\usepackage{physics}
\usepackage[misc,geometry]{ifsym} 
\usepackage{algorithm}
\usepackage[noend]{algpseudocode}
\makeatletter
\newcommand{\@chapapp}{\relax}%
\makeatother

\usepackage[title,toc,titletoc,header]{appendix}

\begin{document}
\title{On Post-Processing the Results of Quantum Optimizers }
%
%
\author{First Author\inst{1}\orcidID{0000-1111-2222-3333} \and
Second Author\inst{2,3}\orcidID{1111-2222-3333-4444} \and
Third Author\inst{3}\orcidID{2222--3333-4444-5555}}

\author{Ajinkya Borle\textsuperscript{(\Letter)} \and
Josh McCarter}
\authorrunning{Borle and McCarter}
\institute{CSEE Department,University of Maryland Baltimore County, Baltimore MD 21250
\email{\{aborle1,jmccar1\}@umbc.edu}\\
}
\maketitle              
\begin{abstract}
The use of quantum computing for applications involving optimization has been regarded as one of the areas it may prove to be advantageous (against classical computation). To further improve the quality of the solutions, post-processing techniques are often used on the results of quantum optimization. One such recent approach is the Multi Qubit Correction (MQC) algorithm by Dorband. In this paper, we will discuss and analyze the strengths and weaknesses of this technique. Then based on our discussion, we perform an experiment on how pairing heuristics on the input of MQC can affect the results of a quantum optimizer and a comparison between MQC and the built-in optimization method that D-wave Systems offers. Among our results, we are able to show that the built-in post-processing rarely beats MQC in our tests. We hope that by using the ideas and insights presented in this paper, researchers and developers will be able to make a more informed decision on what kind of post-processing methods to use for their quantum optimization needs.

\end{abstract}
\section{Introduction}
We are entering the era of Noisy Intermediate Scale Quantum (NISQ) devices \cite{preskill2018quantum}, as of the time of writing this paper. But these devices may not be fault-tolerant to run the traditional quantum algorithms (like Shor's or Grover's Algorithm \cite{shor1994algorithms,grover1997quantum}) for doing computation on a useful scale. However, applications such as quantum chemistry \cite{peruzzo2014variational}, sampling \cite{adachi2015application} and optimization \cite{farhi2014quantum,neukart2017traffic} among others, are the first to make use of such devices.

It is important to understand that when we talk about NISQ devices, we are also considering quantum annealers such as the D-wave 2000Q to be in that category. This is because, as Preskill points out in his work \cite{preskill2018quantum}, the quantum annealer is a noisy implementation of adiabatic quantum computing. While there is still controversy about the lack of conclusive evidence of a quantum speedup, research has highlighted areas of promise \cite{adachi2015application,o2018nonnegative,o2018approach,harris2018phase}.

For the scope of this work, our domain of interest is quantum optimization. In particular, it is the post-processing that is applied on the results returned by  quantum optimizers. The basic hypothesis is that \cite{dorband2018method}, even if quantum devices cannot reach the global minimum for a hard problem, it can still reach the neighborhood of such a solution. Thus, post-processing the output of quantum solvers (irrespective of the type) can be helpful to find an improved solution at the very least, if not the best one.

Our aim in this paper is to study some prominent post-processing techniques used in the field, with a special focus on Multi Qubit Correction (MQC) \cite{dorband2018method}. Then based on our study, we perform some experiments on it. Section 2 covers the required background information. Section 3 deals with the review of some of the most prominent post-processing techniques. In section 4, we discuss and theoretically analyze MQC. Based on what we learn in section 4, we perform experiments in section 5 on how the order of inputs given to MQC can affect the final result of the optimization. We lay out future work based on those empirical results in section 6. Finally, we end with concluding remarks in section 7.

The techniques discussed and proposed in this paper will focus around results from the D-wave quantum annealer. However, they are not limited to the D-wave (except for reverse annealing) and can be applied on results of Quantum Approximation Optimization Algorithm (QAOA) \cite{farhi2014quantum}, Coherent Ising Machines (CIM) \cite{inagaki2016coherent}, Quantum Inspired Digital Annealers (QIDA) \cite{aramon2018physics} and even various classical optimizers based on the Ising model.

\section{Background}
In this section, we shall lay out the terms and concepts we will use in the rest of the paper.
\subsection{The Ising Model}
The Ising Model is a mathematical model originally used in statistical mechanics for ferromagnetism \cite{gallavotti2013statistical}. However, it has applicability beyond statistical mechanics, especially for modeling NP-Hard problems. Quantum Annealers \cite{kadowaki1998quantum} and gate-based optimization approaches like QAOA \cite{farhi2014quantum} are also based on the Ising Model. The two dimensional Ising model, on which the D-wave 2000Q is based, has the following objective function :
\begin{align}
    F(h,J) = \sum_{a}h_a\sigma_a + \sum_{a<b}J_{ab}\sigma_{a}\sigma_{b}\label{eq:ising}
\end{align}

where $\sigma_{a}$ is a binary variable which can take either $-1$ or $+1$,  $h_a$ and $J_{ab}$  \cite{dorband2016stochastic} are the coefficients for the linear and quadratic terms respectively. The $\sigma_{a}$'s binary variables are mapped to qubits in a quantum computer. The quantum optimizer's job is to return the set of values for $\sigma_{a}$s that would correspond to the smallest value of $F(h,J)$ (or the largest value for QAOA).

\subsection{Quantum Annealing}
The Quantum Annealing process uses quantum mechanics to search the energy landscape of the Ising model to find the ground state configuration of $\sigma_a$ variables from Eqn(\ref{eq:ising}). The $\sigma_a$ variables are called as qubits spins in quantum annealing, essentially being quantum bits.

The process begins with the qubits in equal quantum superposition: which means that at this stage, all the potential qubit configurations have an equal probability of being measured. It then attempts to find the lowest energy configuration of the objective function $F(h,J)$ by varying the tunneling field strength (and gradually reducing it to 0), generating the stochastic distribution proportional to $e^{-F(h,J)}$. Under some conditions, these devices can sample from  thermal distribution proportional to $e^{-\beta F(h,J)}$ where $\beta$ is the inverse temperature parameter \cite{tanaka2017quantum}. Here the tunneling field plays a similar role to that of the `temparature' parameter in simulated annealing. It is not yet clear if the D-wave quantum annealer adheres to the adiabatic principle completely (mostly due to technical constraints and noise). A more detailed description \cite{farhi2000quantum} can be found in the book by Tanaka et al. \cite{tanaka2017quantum}. 

For our purposes however, we are interested mainly in the results that the quantum annealer provides us. From an accuracy perspective, a quantum annealer is essentially trying to take samples of a Boltzmann distribution whose energy is the Ising objective function  \cite{adachi2015application}
\begin{align}
    P(\sigma) &= \frac{1}{Z} e^{-F(h,J)}\label{eq:prob_boltzmann_ising}\\
   \text{where } Z &= \textit{exp}\big(\sum_{\{\sigma_{a}\}}\big[ \sum_{a}h_a\sigma_a + \sum_{a<b}J_{ab}\sigma_a\sigma_b \big] \big)\label{eq:prob_partition_ising}
\end{align}
Eqn(\ref{eq:prob_boltzmann_ising}) tells us that the qubit configuration of the global minimum would have the highest probability to be sampled. Because the quantum annealer is a probabilistic machine, we run it multiple times to get a set of solutions. The run (a configuration of values for the variables in the problems) with the lowest energy is taken as the final result. Alternatively, these runs can be fed into a post-processing method in the hopes of getting a better result.
\subsection{Quantum Approximate Optimization Algorithm (QAOA)}
This is a hybrid quantum/classical algorithm that gained popularity due to its simple approach \cite{farhi2014quantum}, can run well on NISQ devices \cite{preskill2018quantum} and is a candidate for being able to achieve `quantum supremacy' \cite{farhi2016quantum}, a term coined for a situation where a quantum computer would be able to do a task on a scale that's not practical on classical computers \cite{preskill2012quantum}.

For the basics of gate-based quantum computation, we recommend the text by Nielsen and Chuang \cite{nielsen2002quantum}. The QAOA algorithm requires the  Hermitian matrix $C$, a Hamiltonian of of Eqn(\ref{eq:ising}) \cite{amin2018quantum} and another Hermitian matrix $B$ represented as
\begin{align}
    C &= \sum_{a}h_a\sigma_{a}^{(z)} + \sum_{a<b}J_{ab}\sigma_{a}^{(z)}\otimes\sigma_{b}^{(z)}\label{eq:ising_hamil}\\
    B &= \sum_{a} \sigma_{a}^{(x)}\\
    \text{where }\sigma^{(z)} &=
    \begin{pmatrix}
    1 & 0\\
    0 & -1
    \end{pmatrix} \text{ and } \sigma^{(x)} =
    \begin{pmatrix}
    0 & 1\\
    1 & 0
    \end{pmatrix}
\end{align}
Each $\sigma_{a}^{(z)}$ and $\sigma_{a}^{(x)}$ can be represented as 
\begin{align}
    \sigma_{a}^{(z)} = (\otimes_{i=1}^{a-1}I) \otimes (\sigma^{(z)}) \otimes (\otimes_{i=a+1}^{N}I)\label{eq:sigmaz}\\
    \sigma_{a}^{(x)} = (\otimes_{i=1}^{a-1}I) \otimes (\sigma^{(x)}) \otimes (\otimes_{i=a+1}^{N}I)\label{eq:sigmax}
\end{align}
The algorithm is based on the Hamiltonian simulation of $B$ and $C$. The strength of this algorithm is that it uses a set of $2p$ angles in order to optimize (maximize) the objective function (where $p$ can be chosen to be far lesser than $N$). The classical part of the algorithm is to search those angles. These can be done by various techniques like grid search, bayesian optimization etc. For more detailed information, we recommend Farhi et. al's paper \cite{farhi2014quantum}.
 
\begin{algorithm}
\caption{Quantum Approximate Optimization Algorithm}\label{alg:qaoa}
\begin{algorithmic}[1]
\Procedure{MAIN}{$B,C,p$}\Comment{The main routine of the algorithm}
\State $\beta \gets \{\emptyset\}, \gamma \gets \{\emptyset\},\beta_1 \gets \emptyset,\gamma_1 \gets \emptyset, config\_res \gets \{\emptyset\}, res \gets \{\emptyset\}$

\State Pick at random $\beta \in [0,\pi]^p,\gamma_{i} \in [0,2\pi]^p$
\While{$(\beta,\gamma)$ can be further optimized, or a limit is reached}
    \For{\texttt{a fixed number of iterations}}
        \State $res \gets res$ $\cup$ QAOA($B,C,\beta,\gamma,p$)
    \EndFor
    \State Tally $res$ and put most occurring result $r \in res$ into $config\_res$
    \State Based on the $config\_res$, pick new $2p$ angles $(\beta,\gamma)$ by classical optimization
\EndWhile
\State \textbf{return} $config\_res$\Comment{This is the approx. soln}
\EndProcedure
\Procedure{QAOA}{$B,C,\beta,\gamma,t$}\Comment{The quantum procedure}
\State Initialize $N$ qubits, $\ket{\psi} \gets \ket{0}^{\otimes N}$
\State Apply Hadamard transform, $\ket{\psi} = 1/\sqrt{2^N}(\ket{0} + \ket{1})^{\otimes N}$
\State $j \gets 1$
\While{$j \leq t$}\Comment{Apply uptil $\beta_t,\gamma_t$}
\State $\ket{\psi} \gets e^{-i\gamma_{j}C}\ket{\psi}$
\State $\ket{\psi} \gets e^{-i\beta_{j}B}\ket{\psi}$
\State {$j \gets j + 1$}
\EndWhile
\State Measure $\ket{\psi}$ in standard basis and store in a classical register $o$
\State \textbf{return} $o$
\EndProcedure
\end{algorithmic}
\end{algorithm}
\section{Post-processing techniques : a review}
In this section we shall review a few prominent post-processing techniques. This will help the theoretical analysis in the next section.
\subsection{Built-in Optimization Postprocessing}
The D-wave developer guide offers a optimization postprocessing technique \cite{dwavepp} that is based on heuristics to decompose the problem graphs \cite{markowitz1957elimination} (either the native hardware graph or a logical graph) into several low treewidth graphs. Then each of these subgraphs are solved locally based on belief propagation on junction trees  \cite{jensen1996bayesian} in the hopes of getting a better solution.
\begin{algorithm}
\caption{Built-in post processing}\label{alg:built_in_pp}
\begin{algorithmic}[1]
\Procedure{MAIN}{$R,G$}\Comment{R is the set of results, G is the graph}
\State Initialize $R' \gets \{\emptyset\}$
\State Decompose graph $G$ into a set of subgraphs $G'$
\For {each run r in R}
    \For {each subgraph $g$ in $G'$}
        \State Use belief propagation on junction trees to optimize $r$ for subgraph $g$
    \EndFor 
    \State Put locally optimized $r$ in the solution set $R'$
\EndFor
\State \textbf{return} $R'$
\EndProcedure
\end{algorithmic}
\end{algorithm}

It is also important to mention that the D-wave API offers a sampling post-processing technique \cite{dwavepp} to create an approximate Boltzmann distribution for a user defined inverse temperature parameter $\beta$. However, since the focus of this work is enhancing quantum optimization, we are not going into the details of such a technique.
\subsection{Multi Qubit Correction (MQC)}
\begin{figure}
\includegraphics[scale=0.30]{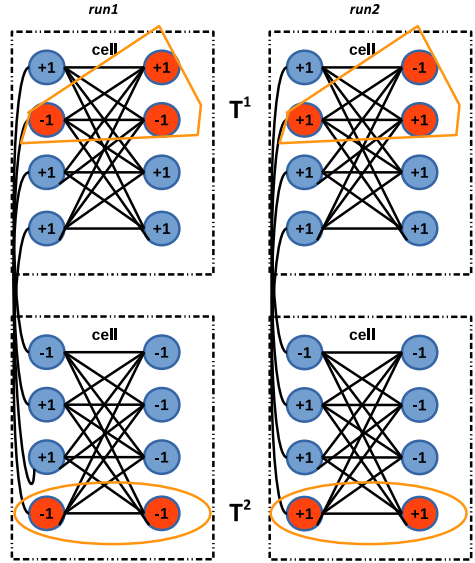}
\centering
\caption{Example of Tunnels across two runs. We can see tunnels $T^1$ and $T^2$ formed in the above configurations for the corresponding qubit values that don't agree (in red). For the problem graph in the example, we take two cells from the chimera graph arranged vertically.} \label{fig:mqc_tunnel}
\end{figure}
\begin{algorithm}
\caption{Multi Qubit Correction (MQC)}\label{alg:mqc}
\begin{algorithmic}[1]
\Procedure{MAIN}{$R,G$}\Comment{R is the set of results, G is the graph}
\While{$|R| \neq 1$}
    \State $R' \gets \{\emptyset\}$
    \State Pair runs from $R$ in put in set $P$
    \For {each run pair $(r,s)$ in P}
        \State $R' \gets R' \cup $ MQC($r,s,G$)
    \EndFor
    \State $R \gets R'$
\EndWhile
\State \textbf{return} $R$
\EndProcedure
\Procedure{MQC}{$run1,run2,G$}
\State Initialize $ans\_run \gets \{\emptyset\}$
\State Create an index set $S$ for qubits that have the same value across $run1$ and $run2$ 
\State Create an index set $D$ for qubits that disagree in value across $run1$ and $run2$
\State $ans\_run \gets ans\_run \cup \{r_i:r_i\in run1 \vee r_i \in run2, i \in S\}$
\State Find connected components in $D$ to create subsets  $T^1,T^2...T^k$ 
\State $i \gets 1$
\While{$i \leq k$}
    \State Compare energy contribution of $T^i$ w.r.t run1 and run2 by Eqn \ref{eq:ising_contrib}
    \State Select configuration of $T^i$ that has lower energy and add in $ans\_run$
    \State $i \gets i + 1$
\EndWhile
\State \textbf{return} $ans\_run$
\EndProcedure
\end{algorithmic}
\end{algorithm}

In 2018, the Multi Qubit Correction (MQC) technique by Dorband \cite{dorband2018method} was proposed as a simple,fast and effective technique to improve upon the results of an Ising problem optimizer (focused on, but not limited to the D-wave quantum annealers). The MQC technique is especially beneficial when an optimizer returns solutions near the actual global minimum or has components of the global minimum.

The technique works by pairing runs $run1$ and $run2$. It then makes a set of indices of qubits $D$ that have different values across the two runs and a set $S$ for those that have the same values. Then within the set $D$, we find all $T^i$ subsets of qubits (also known as tunnels) that are transitively connected to each other, i.e $D = \{T^1,T^2,...T^k\}$ where $k$ is the total number of connected components. Figure 1 shows an example of tunnels being formed for a pair of runs. In other words, we find all the connected components in $D$. No qubit from a subset $T^i$ is connected to a qubit in $T^j$ ($i \neq j$). Once that is done, we look at the relative energy each $T^i$ contributes to the global energy. This is done by
\begin{align}
    I^{i}_{run1}(h,J) = \sum_{a\in T^{i}}h_a\sigma_a^{(run1)} + \sum_{a \in T^i}\sum_{b \in S}J_{ab}\sigma_{a}^{(run1)}\sigma_{b}\label{eq:ising_contrib}
\end{align}
Eqn(\ref{eq:ising_contrib}) is the energy contribution by the tunnel $T^i$ for the configuration of $run1$. A Similar equation can be created for the energy contribution of $run2$ for the same tunnel. We then select the configuration (from the two paired runs) that contributes the lowest of the two. This is done for each tunnel.

The $N$ runs obtained from the quantum optimizer are paired for applying MQC. The result of step above are $N/2$ new `runs' that go through the same procedure. In this manner, $N$ runs are reduced down to 1 run.

\subsection{High Precision Enhancement (HPE)}
\begin{algorithm}
\caption{High Precision Enhancement (HPE)}\label{alg:hpe}
\begin{algorithmic}[1]
\Procedure{MAIN}{$h,J,L,n$}\Comment{$h$ and $J$ are Ising coefficients, $L$ is the set of scaling factors, $n$ is the number of runs for the set $R^l$ where $l\in L$}
\State Intialize $F\_set \gets \{\emptyset\}, final\_ans \gets \emptyset$
\For {$l \in L$}
    \State Scale $F(h,J)$ with $l$
    \State Run the scaled problem on a quantum optimizer n times and store runs in $R^l$
\EndFor

\State $i \gets 1$
\While {$i \leq n$}
    \State $T \gets \{\emptyset\}$
    \For {$l \in L$}
        \State Pick a run $r^l$ from $R^l$ and store in $T$ 
    \EndFor
        \State Apply MQC on the runs in $T$ to reduce it to 1 run, store in $F\_set$
        \State $i \gets i + 1$
\EndWhile
\State Apply MQC on $F\_set$ to reduce it to 1 run, store in $final\_ans$
\State \textbf{return} $final\_ans$
\EndProcedure
\end{algorithmic}
\end{algorithm}
Another work by Dorband in 2018 proposed a specialized post-processing technique for Ising Problems that have higher precision requirements than what the quantum hardware supports  \cite{dorband2018extending}. This is done by scaling the problem $F(h,J)$ by $l,l\in L$ where $L$ is a set of scaling factors. This is because, although the various scaled problems have the same solution configuration for their global minimums, the quantum hardware would treat them as different problems due to the limitations of precision for coefficients.

Each of the $|L|$ scaled problems are then fed into the quantum optimizer to obtain the set $R^l$ (results from the optimizers) where $l \in L$ . Then we pick one run for each of the $|L|$ problems (from the $R^l$s) and reduce it to 1 run using MQC. We repeat this procedure till we cover all $|L| \times n$ runs (where $n = |R^l|, \forall l \in L$). Finally, when we are left with the aggregated $n$ runs (from the previous $|L|$ iterations), we use MQC to systematically aggregate it into the final run, which is our result. For more details, we recommend the original work to the readers  \cite{dorband2018extending}.

In this way, the classical postprocessing overcomes the lack of precision in the quantum optimizer, without having to work on all the variables (qubits). It is also important to note that HPE is a different type of post-processing technique than the others discussed in this paper. This is because it serves a different purpose than the other techniques and can be seen as a meta-technique that can incorporate other post-processing methods as a subroutine.

\subsection{Other post-processing techniques}
\subsubsection{Reverse Annealing: } This is a feature of the D-wave annealers introduced in 2017 with the D-wave 2000Q \cite{revanneal2017,chancellor2017modernizing,ohkuwa2018reverse}, that allows for the usage of a previous run of the D-wave as a starting point for new runs. In its initial form, it was first proposed by Perdomo-Ortiz et al. \cite{perdomo2011study}. The utility value of doing such a thing is based on the conjecture that even if the global minimum has not been found by the quantum annealer, the annealer's solution can act as a good starting point to search for better solutions in the neighborhood.

Instead of starting off with equally weighted qubit variables, we start of with the solution and work backwards to a mid-anneal superposition. After that the forward anneal starts again in order to search for another (hopefully better) solution. The search for a suitable mid-aneal superposition is controlled by the \textit{reversal distance} parameter.

\subsubsection{Sample Persistence: } In 2016, Karimi and Rosenberg proposed a technique to improve the results received from a quantum annealer \cite{karimi2017boosting}. It involves fixing the qubits whose values stay the same across the various runs (within a threshold) when retrieved from the quantum annealer, and then run a subset of the original problem (for qubits that are not fixed). This would require us to modify the subset of the problem such that the $J$ coefficients of the couplings that connect to the fixed variables are added (or subtracted, as per the sign of the fixed variable) to the $h$ values of the qubits on the other end of those couplings (i.e. the ones inside the subset).

The idea behind this approach is the conjecture that it is easier to solve a smaller subset of the problem with a greater chance of success than it is to solve the complete problem in one go. It also assumes that the qubits that show the same configuration across multiple runs are more likely to be the correct values for the global minimum solution and thus, they are fixed. More details about their work can be found in their papers  \cite{karimi2017boosting,karimi2017effective}

\subsubsection{Ochoa et al.'s technique: } Recently, we came to know of the work done by Ochoa et al. \cite{ochoa2019feeding} for improving the sampling done by a quantum annealer. The aim of this approach is to arrive at lower energy samples (or runs) (compared to the set of samples/runs that it begins with) during its polynomial time sampling procedure. In contrast, the MQC procedure follows a greedy descent type of approach. A comparison between these two would be an interesting future work.
\section{Analysis and Discussion of MQC}
\subsection{On the time complexity of MQC}
Although it is not guaranteed that MQC would reach the global minimum, it is still important to consider the amount of time the entire post-processing would require. In the work by Dorband  \cite{dorband2018method}, it was stated that $N$ runs can be aggregated to $1$ run in $O(\text{ceil}(\log N))$ aggregation steps. While this is true, its not the complete time complexity, for which we must analyze the computation within each run.

Let $q$ be the total number of qubits in the problem. Let $D_{max}$ be the largest set of differing qubits that would be encountered when the N runs are being aggregated into 1 run. Each pair would require
\begin{enumerate}
    \item A linear search to see whether the qubits match or differ : $q$ or $|V|$, since $V$ is the set of all the vertices in our graph (this can either be the logical problem graph or the hardware graph like the chimera),
    \item A connected component analysis using DFS to find the transitive connectivity : $O(|V_{Dmax}| + |E_{Dmax}|)$. We can simplify this to $O(|V|+|E|)$ since they are bounded by the total number of vertices and edges. 
    \item A check of each tunnel $T^{i}$ for the relative energy it contributes: $O(|V_{T^i}| +|E_{T^i}^{\text{out}}|)$ 
    (this is for calculating the energy associated with Eqn(\ref{eq:ising_contrib})) $E_{T^i}^{\text{out}}$ is the edge set that connects qubits in $T^i$ with qubits in the $S$ set. We do this for k tunnels. So we can see that :
    \begin{align}
        \begin{split}
            \text{Cost of calculating energy} \sim O(|V_{T^1}|+|V_{T^2}|+... + |V_{T^k}|\\ + |E^{out}_{T^1}|+|E^{out}_{T^2}|+...+|E^{out}_{T^k}|)\\
            \text{or, } O(|V_{D_{max}}| + |E^{out}_{D_{max}}|) \sim O(|V|+|E|)
        \end{split}
    \end{align}
    where $|V_{D_{max}}|$ and $|E^{out}_{D_{max}}|$ is the summation of all the vertices and edges (that connect to set $S$) that are in $D_{max}$.
\end{enumerate}
Thus each step takes about $O(|V| + |E|)$ operations. 
\begin{align}
    \text{Cost per pair} \sim O(|V|+|E|)  \label{eq:cost_pair}
\end{align}
The total number of pairs to process go down by half in each step. Thus it is
\begin{align}
    \text{Total Number of pairs} = N/2 + N/4 +N/8 ... + 1 \label{eq:tot_pair}
\end{align}
In big O notation, $O(N/2 + N/4 +N/8 ... + 1) \sim O(N)$
Thus the cost of reducing $N$ runs down to 1 run using MQC is
\begin{align}
    \text{Total Cost} \sim O((|V|+|E|)N)
\end{align}
For a sparse graph like the D-wave's Chimera, the complexity will simplify to $O(|V|N)$. However, as the graph approaches full connectivity (i.e. $|E|\rightarrow \binom{|V|}{2} $), the complexity will go up to $O(|V|^2N)$. But it should be noted that MQC will have a harder time with denser graphs and become totally ineffective when the graphs become fully connected (explained in the next section).
\subsection{MQC is ineffective for fully connected graphs}
The MQC technique relies upon two conditions for it to be effective :
\subsubsection{Condition 1: } \textit{At least two tunnels need to be formed for a given pair of runs, $run1$ and $run2$}
\\
\textbf{Condition 2: } \textit{The configuration of qubits that contribute the lowest energy for the tunnels shouldn't be all from $run1$ or $run2$ exclusively}
\\
Hence, if for a pair of runs, only a single tunnel is formed, irrespective of the run pair, then selecting either the tunnel configuration from $run1$ or $run2$ would make no difference as it would be equivalent of selecting the entire configuration of $run1$ or $run2$.

\begin{theorem}
If the graph of the Ising Problem in question is fully connected. Then the MQC algorithm will not be able to optimize on the set of runs $R$ it receives from the Ising solver. 
\end{theorem}
\begin{proof}
A fully connected graph will have every qubit (or vertex) be connected with every other qubit in the graph. Thus even when we have a set of qubits $D$ that differ across $run1$ and $run2$, we won't be able to find multiple tunnels since there will be a single connected component in the subgraph. In other words, since each qubit is connected with every other qubit in the set of all qubits $S \cup D$ or $V$, the qubits within $D$ are also fully connected. This leaves us with a single tunnel. Since Condition 1 is a necessary condition for MQC to do optimization, this will result in a failure to optimize the runs received.
\end{proof}

\subsection{The result of MQC can depend on how the runs are paired} \label{sec:pair}
In the general case of the Ising problem, the way runs are paired together can affect the final result of the MQC algorithm. In other words, the result of MQC is not independent of the initial pair configuration of the runs retrieve from the Ising solver.

Given a set of tunnels T = $\{T^1_{(1,2)},T^2_{(1,2)}...T^k_{(1,2)}\}$ for two runs $r_1$ and $r_2$, it is important for us to understand that when we optimize for each tunnel (by selecting which energy is lower amongs $r_1$ and $r_2$ for that tunnel), we are essentially doing local optimization.

\textbf{Conjecture} : \textit{In the general case, MQC can produce a different final result if the initial pairing is done differently, for a given set of runs $R$ in the general case Ising problem.}\\
 As we aggregate the runs down to the final result, the possibility of encountering multiple other tunnels being formed over the problem graph (logical or hardware) is high. This can be thought of as a sequence of local optimizations occurring  as and when the tunnels get formed. 

Now if we form the pairs of runs from $R$ in a different way than we do before, we may get a different set of tunnels being formed over the problem graph as the runs get reduced. This opens a possibility of optimizations happening differently, leading up to a different result. The reason for this is because MQC decides the configuration of a tunnel with respect to the contribution to the global energy (i.e with respect to what is outside the tunnel). In other words, the configuration of qubits in $S$ also influences the choices made for each tunnel by MQC. Even if we get to encounter the same kinds of tunnels in the $\log N$ steps, the possibility that they may appear in a different sequence would be equivalent of doing a different sequence of local optimizations. Which opens a possibility of a different final result.

It is important to note that this argument is solely based upon the working of the MQC technique. It does not take into consideration what optimizer is attached to it. There might still exist proofs for the pairing order of the input not affecting the result, when it comes to specific optimizers.

\subsection{Discussion}\label{sec:discussion}
Based on the theoretical analysis of MQC done above, we would like to make a few other comments on it. These points may prove useful for future work.

\subsubsection{MQC relies on a lot of unique samples to be effective:} Because the core of MQC essentially requires choosing between two configurations that are mirror opposites for each tunnel, it is safe to say that we miss out on a more nuanced approach for optimizing these tunnels. Thus, because of the simplicity of the technique and the fact that solutions don't get worse as they are processed \cite{dorband2018method}, more sampled runs would mean more opportunities for improvements of the result.

\subsubsection{Sample Persistence and MQC:} Both MQC and Sample Persistence work on two common principles : (i) the concept of fixed variables/qubits based on which qubits have the same values across runs and (ii) the qubits that have differing values across runs form `tunnels', that need to be optimized. The difference arises in how they treat the set of runs received from the quantum optimizer, and how they treat the tunnels. Where MQC is a purely classical technique, the Sample Persistence method believes in using the quantum annealer multiple times, in order to resolve the tunnels.  

Thus, in theory one can expect a better solution in the case of Sample Persistence, since it does more than just compare between two configurations of a set of qubits. Maybe its even possible that it is more effective than MQC for a smaller number of samples. However, Sample Persistence is a highly parameterized technique \cite{karimi2017boosting}, this brings about a different set of problems as different parameters need to be tested out in order for the technique to be effective. A thorough empirical comparison between MQC and Sample Persistence would be beneficial for us to assess which of these techniques is better. Unfortunately, this falls outside the scope of this paper but we would like to suggest it as a future work.

Even if Sample Persistence turns out to be more effective than MQC, a case can still be made for MQC because it is a polynomial time algorithm that requires classical resources. Sample Persistence requires additional calls to be made to the quantum optimizer. Thus, in the case where cost of classical resources is cheaper than the cost of accessing quantum resources (in terms of monetary value), it may be economical to use MQC over Sample persistence.

\subsubsection{Sample persistence cannot be used for HPE:} Despite the advantages that Sample Persistence may have over MQC, it cannot be used as a subroutine for HPE.  This is because HPE depends on a post-processing that requires a higher precision capability than the quantum optimizer that processes the result. This is possible (and convenient) by using classical computation for the post-processing.

\section{Experiment}
Based on the analysis and discussion in the previous section, we want to empirically observe how pairing schemes might affect MQC's final result. By using the information of section \ref{sec:pair}, we can choose to consciously pair the runs in a particular manner. For the purposes of our experiment, we will use two heuristics to pair the input runs:
\begin{enumerate}
    \item On the basis of similar energy (Rank Ordering)\label{list:heu1}
    \item On the basis of difference in qubit values (Maximum Difference)\label{list:heu2}
\end{enumerate}

The inspiration for the first heuristic is the evolutionary computational approach  \cite{greenwood2001finding}, where the fitness of the individuals (or runs) inside the population (or the set of runs) is evaluated and the best-fit individuals are selected for reproduction (or form an input for MQC in our case). 

The logic behind the second heuristic is based on the property of MQC that the output cannot be worse than both runs in the input pair. Hence, we pair runs in a way that are the most different from each other, based on the values of qubits. In this way, we hope to extract useful tunnels that may help improve the results of MQC. It should be noted however that this heuristic will cost an additional $O(|V|N^2)$ (${N \choose 2}$ comparisons and $|V|$ per comparison) for each of the $N$ runs.

Another objective of this experiment is to compare MQC with the built-in post-processing technique that the D-wave API offers. This has not been done in any of the related works till date.
\subsection{Experiment setup}
For this experiment, we create a set of 50 different problems based on the Ising objective function for the on-chip graph of a D-wave 2000Q machine. This is not the lower noise machine released to the public in May of 2019 \cite{dwave25lower}. We used the OCEAN SDK with Python 2.7 for this task. Each one of our 50 problems utilizes all the 2038 available qubits of our D-wave solver: \texttt{DW\_2000Q\_2\_1}. The coefficients are created by random uniform sampling (seeded to 316) in the range of $[-2,2]$ for $h$ coefficients and $[-1,1]$ for $J$s. It should be noted that random chimera graph problems like these may not be considered as problems against which quantum speedup could be achieved \cite{katzgraber2014glassy}. However, these experiments are about studying the behavior of MQC and its variants on the results of quantum annealing.

The above problems are annealed in three modes : (a) without any post-processing, (b) sampling post-processing mode and (c) optimization post-processing mode. The reason for using the sampling post-processing is to create a more diverse range of solutions in terms of configuration and energy. This is done to see how the various versions of MQC perform on a set of runs that are not extremely similar to each other.

The results from (a) and (b) are then run through standard MQC, MQC with heuristic \ref{list:heu1}(Rank Ordering) and MQC with heuristic \ref{list:heu2} (Maximum Difference).

Each of above 3 annealing operations is done for obtaining results for 1000 and 2000 runs of the machine. Each run has an anneal time of $20 \mu s$. 
\subsection{Results and discussion}
\begin{table}[ht]
\centering
\caption{Comparison between MQC done for raw and sampling (smpl) results }
\label{table:1}
 \begin{tabular}{|c|c|c|c|}
 \hline
 \multirow{2}{*}{\textbf{Runs}} & \multicolumn{3}{|c|}{mqc on raw vs mqc on smpl } \\
 \cline{2-4} & raw $=$ smpl & raw $<$ smpl & raw $>$ smpl\\ 
 \hline
 1000 & 31 & 18 & 1\\
 \hline
 2000 & 40 & 8 & 2\\
\hline
\end{tabular}
\end{table}

\begin{table}[ht]
\centering
\caption{Comparison between MQC and built in optimization post-processing }
\label{table:2}
 \begin{tabular}{|c|c|c|c|c|}
 \hline
 \multirow{2}{*}{\textbf{Runs}} & \multirow{2}{*}{\textbf{Mode}} &
 \multicolumn{3}{|c|}{mqc vs pp (built-in opt. pp)}\\
 \cline{3-5} &
 & mqc$=$pp & mqc$<$pp & mqc$>$pp\\
 \hline
 1000 & raw & 4 & 46 & 0\\
 \hline
 1000 & sampling & 5 & 42 & 3\\
\hline
2000 & raw & 9 & 41 & 0\\
\hline
2000 & sampling & 9 & 38 & 3\\
\hline
\end{tabular}
\end{table}

\begin{table}[ht]
\centering
\caption{Comparison between standard MQC and MQC with pairing heuristics}
\label{table:3}
 \begin{tabular}{|c|c|c|c|c|c|c|c|}
 \hline
 \multirow{2}{*}{\textbf{Runs}} & \multirow{2}{*}{\textbf{Mode}} &
 \multicolumn{3}{|c|}{mqc vs rnk (Rank Order MQC)} &
 \multicolumn{3}{|c|}{mqc vs mdf-mqc( Max. Diff MQC)}\\
 \cline{3-8} &
 & mqc$=$rnk & mqc$<$rnk & mqc$>$rnk &
 mqc$=$mdf & mqc$<$mdf & mqc$>$mdf\\
 \hline
 1000 & raw & 48 & 1 & 1 & 48 & 1 & 1\\
 \hline
 1000 & sampling & 34 & 5 & 11 & 27 & 18 & 5\\
\hline
2000 & raw & 50 & 0 & 0 & 47 & 3 & 0\\
\hline
2000 & sampling & 41 & 5 & 4 & 34 & 13 & 3\\
\hline
\end{tabular}
\end{table}

Each value in the tables above is the number of instances or problems for which an energy comparison holds true (as indicated by its respective column).

In our results, there were no problem instances where the raw energies of the D-wave's results were better than the standard MQC's. During our tests, as Table 2 indicates, there were a total of only 6 instances where the (standard) MQC had worse energy than the built-in optimization post-processing. All of these 6 instances were when the sampling mode was used to generate the runs. While in the raw mode, there was no instance where a better solution was derived from built-in post processing over MQC. However, the amount of instances for which MQC has an advantage over the built-in technique drops as we move from 1000 to 2000 runs. This indicates that MQC is more effective when used for fewer runs, though further testing is required. 

Table 1 shows the comparison between MQC done on raw runs and those obtained from the sampling mode. For the most part, the final results of MQC done on raw inputs is equivalent to the results of MQC done on the sampling mode. This number grows as we move from  1000 to 2000 runs. There are a very few cases where MQC done on the sampling mode got a better energy, which is good since operating MQC on the raw results would save computation time as well. 

From our results in Table 3, we can see that neither ranked order nor max difference heuristics are conclusively better than standard MQC for general use. However, the experiment empirically shows that pairing order of the input can have an effect on the output of MQC. This would also indicate that there exists a pairing order that minimizes the end result the most, and it is not evident that the standard MQC is the best way to do so. The results from the sampling mode are more affected by permutation of inputs than the raw results of the D-wave. This means that the raw results of the D-wave are (a) robust against pairing order and (b) very close to each other. This is a good indication of the quality of solutions that D-wave provides.
It will be interesting to see the behavior of MQC when it is used with optimizers other than the D-wave. The results with the sampling mode indicate that MQC would be more sensitive to the pairing order when it receives dissimilar outputs (in this case, approximating a Boltzmann distribution). However, this sensitivity to pairing order seems to diminish as the number of runs are increased. Further testing is recommended.
\section{Future Work}
A thorough comparison of the Sample Persistence \cite{karimi2017effective} technique, Ochoa et al.'s technique \cite{ochoa2019feeding} and MQC would be useful for the research community. Also, as mentioned in the section above, it may be a good idea to test out MQC vs other techniques over larger number of runs. Finally, the effectiveness of MQC also needs to be compared for different types of optimizers : quantum annealers \cite{dwave25lower}, QAOA \cite{farhi2014quantum}, digital annealers \cite{aramon2018physics}, Coherent Ising Machines \cite{inagaki2016coherent} etc.
\section{Concluding Remarks}
In this paper, we first reviewed the prominent post-processing techniques used on the results of quantum optimization. We then theoretically analyzed and discussed the strengths and weaknesses of the Multi Qubit Correction (MQC) technique by Dorband. It was followed by an experiment where we show how the pairing order could effect the final result of the MQC process. We also show that in most instances of our tests, MQC performs better or at par compared to the built-in post-processing technique for optimization. Finally, we outline possible areas of interesting research work that may hold promise when it comes to quantum optimization.

\subsection{Acknowledgment} We would like to thank John Dorband, Milton Halem and Samuel Lomonaco from UMBC for their feedback and discussion on the topics in this paper. We would also like to thank Nicholas Chancellor from Durham University and 
Helmut Katzgraber from Microsoft for their suggestions and feedback. And finally, a thanks to D-wave Systems for providing us access to their machines.
%

%
\bibliographystyle{splncs04}
\bibliography{references}

\begin{thebibliography}{10}
\providecommand{\url}[1]{\texttt{#1}}
\providecommand{\urlprefix}{URL }
\providecommand{\doi}[1]{https://doi.org/#1}

\bibitem{dwave25lower}
D-wave makes new lower-noise quantum processor available in leap,
  \url{https://www.dwavesys.com/press-releases/d-wave-makes-new-lower-noise-quantum-processor-available-leap}

\bibitem{dwavepp}
The d-wave post-processing documentation,
  \url{https://docs.dwavesys.com/docs/latest}

\bibitem{adachi2015application}
Adachi, S.H., Henderson, M.P.: Application of quantum annealing to training of
  deep neural networks. arXiv preprint arXiv:1510.06356  (2015)

\bibitem{amin2018quantum}
Amin, M.H., Andriyash, E., Rolfe, J., Kulchytskyy, B., Melko, R.: Quantum
  boltzmann machine. Physical Review X  \textbf{8}(2),  021050 (2018)

\bibitem{aramon2018physics}
Aramon, M., Rosenberg, G., Miyazawa, T., Tamura, H., Katzgraber, H.G.:
  Physics-inspired optimization for constraint-satisfaction problems using a
  digital annealer. arXiv preprint arXiv:1806.08815  (2018)

\bibitem{chancellor2017modernizing}
Chancellor, N.: Modernizing quantum annealing using local searches. New Journal
  of Physics  \textbf{19}(2),  023024 (2017)

\bibitem{revanneal2017}
D-wave: {Reverse Quantum Annealing for Local Refinement of Solutions}. Tech.
  rep., D-wave Systems (2017)

\bibitem{dorband2016stochastic}
Dorband, J.E.: Stochastic characteristics of qubits and qubit chains on the
  d-wave 2x. arXiv preprint arXiv:1606.05550  (2016)

\bibitem{dorband2018extending}
Dorband, J.E.: Extending the d-wave with support for higher precision
  coefficients. arXiv preprint arXiv:1807.05244  (2018)

\bibitem{dorband2018method}
Dorband, J.E.: A method of finding a lower energy solution to a qubo/ising
  objective function. arXiv preprint arXiv:1801.04849  (2018)

\bibitem{farhi2014quantum}
Farhi, E., Goldstone, J., Gutmann, S.: A quantum approximate optimization
  algorithm. arXiv preprint arXiv:1411.4028  (2014)

\bibitem{farhi2000quantum}
Farhi, E., Goldstone, J., Gutmann, S., Sipser, M.: Quantum computation by
  adiabatic evolution. arXiv preprint quant-ph/0001106  (2000)

\bibitem{farhi2016quantum}
Farhi, E., Harrow, A.W.: Quantum supremacy through the quantum approximate
  optimization algorithm. arXiv preprint arXiv:1602.07674  (2016)

\bibitem{gallavotti2013statistical}
Gallavotti, G.: Statistical mechanics: A short treatise. Springer Science \&
  Business Media (2013)

\bibitem{greenwood2001finding}
Greenwood, G.W.: Finding solutions to np problems: Philosophical differences
  between quantum and evolutionary search algorithms. In: Evolutionary
  Computation, 2001. Proceedings of the 2001 Congress on. vol.~2, pp. 815--822.
  IEEE (2001)

\bibitem{grover1997quantum}
Grover, L.K.: Quantum mechanics helps in searching for a needle in a haystack.
  Physical review letters  \textbf{79}(2), ~325 (1997)

\bibitem{harris2018phase}
Harris, R., Sato, Y., Berkley, A., Reis, M., Altomare, F., Amin, M., Boothby,
  K., Bunyk, P., Deng, C., Enderud, C., et~al.: Phase transitions in a
  programmable quantum spin glass simulator. Science  \textbf{361}(6398),
  162--165 (2018)

\bibitem{inagaki2016coherent}
Inagaki, T., Haribara, Y., Igarashi, K., Sonobe, T., Tamate, S., Honjo, T.,
  Marandi, A., McMahon, P.L., Umeki, T., Enbutsu, K., et~al.: A coherent ising
  machine for 2000-node optimization problems. Science p. aah4243 (2016)

\bibitem{jensen1996bayesian}
Jensen, F.V.: Bayesian updating in causal probabilistic networks by local
  computations. An introduction to Bayesian networks  (1996)

\bibitem{kadowaki1998quantum}
Kadowaki, T., Nishimori, H.: Quantum annealing in the transverse ising model.
  Physical Review E  \textbf{58}(5), ~5355 (1998)

\bibitem{karimi2017boosting}
Karimi, H., Rosenberg, G.: Boosting quantum annealer performance via sample
  persistence. Quantum Information Processing  \textbf{16}(7), ~166 (2017)

\bibitem{karimi2017effective}
Karimi, H., Rosenberg, G., Katzgraber, H.G.: Effective optimization using
  sample persistence: A case study on quantum annealers and various monte carlo
  optimization methods. Physical Review E  \textbf{96}(4),  043312 (2017)

\bibitem{katzgraber2014glassy}
Katzgraber, H.G., Hamze, F., Andrist, R.S.: Glassy chimeras could be blind to
  quantum speedup: Designing better benchmarks for quantum annealing machines.
  Physical Review X  \textbf{4}(2),  021008 (2014)

\bibitem{markowitz1957elimination}
Markowitz, H.M.: The elimination form of the inverse and its application to
  linear programming. Management Science  \textbf{3}(3),  255--269 (1957)

\bibitem{neukart2017traffic}
Neukart, F., Compostella, G., Seidel, C., von Dollen, D., Yarkoni, S., Parney,
  B.: Traffic flow optimization using a quantum annealer. Frontiers in ICT
  \textbf{4}, ~29 (2017)

\bibitem{nielsen2002quantum}
Nielsen, M.A., Chuang, I.: Quantum computation and quantum information (2002)

\bibitem{ochoa2019feeding}
Ochoa, A.J., Jacob, D.C., Mandr{\`a}, S., Katzgraber, H.G.: Feeding the
  multitude: A polynomial-time algorithm to improve sampling. Physical Review E
   \textbf{99}(4),  043306 (2019)

\bibitem{ohkuwa2018reverse}
Ohkuwa, M., Nishimori, H., Lidar, D.A.: Reverse annealing for the fully
  connected p-spin model. Physical Review A  \textbf{98}(2),  022314 (2018)

\bibitem{o2018approach}
O’Malley, D.: An approach to quantum-computational hydrologic inverse
  analysis. Scientific reports  \textbf{8}(1), ~6919 (2018)

\bibitem{o2018nonnegative}
O’Malley, D., Vesselinov, V.V., Alexandrov, B.S., Alexandrov, L.B.:
  Nonnegative/binary matrix factorization with a d-wave quantum annealer. PloS
  one  \textbf{13}(12),  e0206653 (2018)

\bibitem{perdomo2011study}
Perdomo-Ortiz, A., Venegas-Andraca, S.E., Aspuru-Guzik, A.: A study of
  heuristic guesses for adiabatic quantum computation. Quantum Information
  Processing  \textbf{10}(1),  33--52 (2011)

\bibitem{peruzzo2014variational}
Peruzzo, A., McClean, J., Shadbolt, P., Yung, M.H., Zhou, X.Q., Love, P.J.,
  Aspuru-Guzik, A., O’brien, J.L.: A variational eigenvalue solver on a
  photonic quantum processor. Nature communications  \textbf{5}, ~4213 (2014)

\bibitem{preskill2012quantum}
Preskill, J.: Quantum computing and the entanglement frontier. arXiv preprint
  arXiv:1203.5813  (2012)

\bibitem{preskill2018quantum}
Preskill, J.: Quantum computing in the nisq era and beyond. arXiv preprint
  arXiv:1801.00862  (2018)

\bibitem{shor1994algorithms}
Shor, P.W.: Algorithms for quantum computation: Discrete logarithms and
  factoring. In: Foundations of Computer Science, 1994 Proceedings., 35th
  Annual Symposium on. pp. 124--134. Ieee (1994)

\bibitem{tanaka2017quantum}
Tanaka, S., Tamura, R., Chakrabarti, B.K.: Quantum spin glasses, annealing and
  computation. Cambridge University Press (2017)

\end{thebibliography}
\end{document}